\documentclass[sn-basic, Numbered]{sn-jnl}

\usepackage{graphicx}%
\usepackage{subfigure}%
\usepackage{multirow}%
\usepackage{amsmath,amssymb,amsfonts}%
\usepackage{amsthm}%
\usepackage{mathrsfs}%
\usepackage[title]{appendix}%
\usepackage{xcolor}%
\usepackage{textcomp}%
\usepackage{manyfoot}%
\usepackage{booktabs}%
\usepackage{algorithm}%
\usepackage{algorithmicx}%
\usepackage{algpseudocode}%
\usepackage{listings}%

\theoremstyle{thmstyleone}%
\newtheorem{theorem}{Theorem}
\newtheorem{corollary}[theorem]{Corollary}

\theoremstyle{thmstyletwo}%

\theoremstyle{thmstylethree}%

\begin{document}

\title[Article Title]{The airplane refueling problem is NP-complete and is solvable in polynomial time}


\author*[1]{\fnm{Jinchuan} \sur{Cui}}\email{cjc@amss.ac.cn}

\author[1]{\fnm{Xiaoya} \sur{Li}}\email{xyli@amss.ac.cn}


\affil[1]{\orgdiv{Academy of Mathematics and Systems Science}, \orgname{Chinese Academy of Sciences}, \orgaddress{\street{No. 55 Zhongguancun East Road}, \city{Beijing}, \postcode{100190}, \country{China}}}




\abstract{The airplane refueling problem is a nonlinear combinatorial optimization problem, and its equivalent problem the $n$-vehicle exploration problem is proved to be NP-complete \cite{Cui23NP}. In Article \cite{Cui23P}, we designed the sequential search algorithm for solving large scale of airplane refueling instances, and we proved that the computational complexity increases to polynomial time with increasing number of airplanes. Thus the airplane refueling problem, as an NP-complete problem, is solvable in polynomial time when its input scale is sufficiently large.}

\keywords{Airplane refueling problem (ARP), $n$-vehicle exploration problem (NVEP), Combinatorial optimization problem, polynomial-time algorithm, NP-completeness}



\maketitle

\section{Introduction}
\label{intro}

%
%

\qquad The Airplane Refueling Problem (ARP) was raised by Woeginger \cite{woeginger10} from a math puzzle problem \cite{puzzle58}. Suppose there are $n$ airplanes with mid-air refueling technique referred to $A_1, \cdots, A_n$, each $A_i$ can carry $v_i$ tanks of fuel, and consumes $c_i$ tanks of fuel per kilometers for $1\leq i \leq n$. The fleet starts to fly together to a same target without getting fuel from outside, but each airplane can refuel to other airplanes instantaneously during the trip and then be dropped out. The goal is to determine a drop out permutation $\pi = (\pi(1), \cdots, \pi(n))$ that maximize the flight of the last remaining airplane.

Previous research on ARP centralized on its complexity analysis and its algorithm design \cite{vasquez15,gamzu19,lijs19}.  In Article \cite{Cui23P}, we designed the sequential search algorithm for solving large scale of airplane refueling instances, and proved that the computational complexity turns to be polynomial time solvable with increasing number of airplanes. Related work also focused on equivalent problems of ARP such as single machine scheduling problem \cite{hohn13,vasquez17} and the $n$-vehicle exploration problem \cite{lixy09,yu18,zhang21}. In Article \cite{Cui23NP} we proved that NVEP is NP-complete.

\section{Background}
\label{sec:1}

\qquad Let's consider an ordering $\pi$ and its related sequence $A_{\pi(1)} \Rightarrow A_{\pi(2)} \Rightarrow \cdots \Rightarrow A_{\pi(n)}$, in which $A_{\pi(i)}$ refuels to $A_{\pi(j)}$ for $i < j$. As is shown in Fig. \ref{fig:arpmodel} the maximal distance $S_{\pi}$ associated with the ordering $\pi$, and $x_{\pi(i)}$ denotes the distance that $A_{\pi(i)}$ travels farther than $A_{\pi(i-1)}$, which is also the distance that $A_{\pi(i)}$ contributes to the total flight.

\begin{figure}[ht]
  \includegraphics[width=0.65\textwidth]{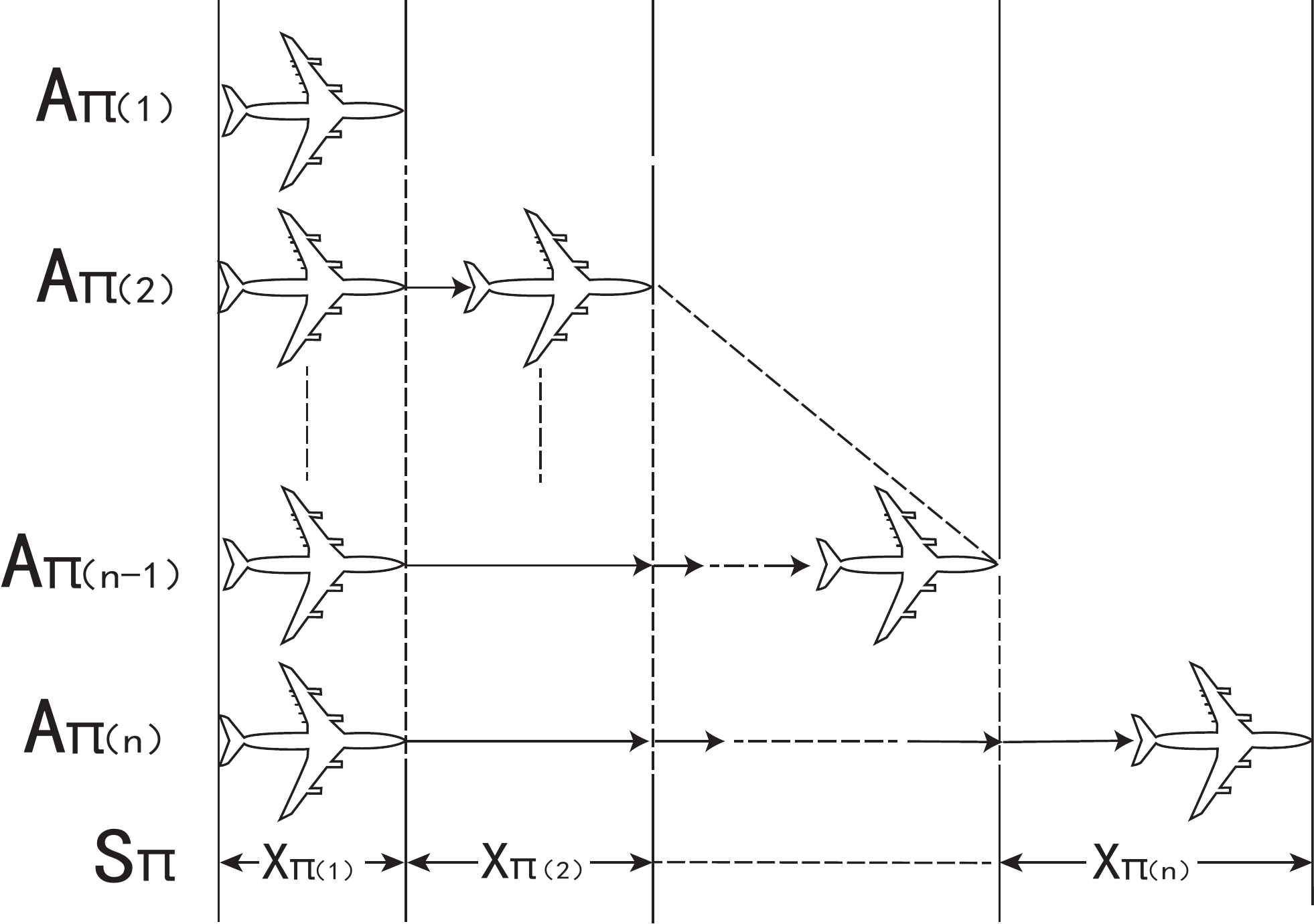}
  \caption{Description of ARP with ordering $\pi$.}
  \label{fig:arpmodel}
\end{figure}

Let $S_{\pi}=\sum\limits_{i=1}^n{x_{\pi(i)}}$ denotes the maximal length that $A_{\pi(n)}$ can approach,  then the goal of ARP is to find a drop out permutation $\pi = (\pi(1), \pi(2), \cdots, \pi(n))$ that maximizes $S_{\pi}$ shown in equation (\ref{eqt1}).


\begin{equation}
S_{\pi} = \frac{v_{\pi(1)}}{c_{\pi(1)}
\cdots + c_{\pi(n)}} + \cdots + \frac{v_{\pi(n)}}{c_{\pi(n)}}\label{eqt1}
\end{equation}

We use $C$ to denote the cumulative sums of fuel consumption rates, induced by an ordering $\pi$. According to equation~(\ref{eqt1}), $C$ is represented as follows.

\begin{equation}
\left( c_{\pi(1)} + \cdots + c_{\pi(n)} \right),  \left( c_{\pi(2)} + \cdots + c_{\pi(n)} \right), \cdots, c_{\pi(n)} \label{eqt2}
\end{equation}

Since each $S_{\pi}$ corresponds to a permutation of $n$ airplanes, there are totally $n!$ of $S_{\pi}$ needs to be compared.

The $n$-vehicle exploration problem is motivated by a problem in a contest puzzle in China \cite{shuxue}: there are $n$ vehicles $V_1, \cdots, V_n$ with fuel capacities $a_i$ and fuel consumption rates $b_i$ for $(1\leq i \leq n)$. It is assumed that these $n$ vehicles start toward the same direction from the same point at the same time. During the trip, they can not get fuel from outside, but at any point any vehicle can stop and transfer its remaining fuel to other vehicles. Similar with ARP, the goal of the NVEP is to determine a drop out permutation $\pi = \{\pi(1), \cdots, \pi(n)\}$ for the vehicles that maximize the traveled length by the last vehicle. Then given an arbitrary drop out order $\pi$, the traveled length of $V_{\pi(n)}$ is shown in (\ref{eqt3}).

\begin{equation}
D_{\pi} = \sum\limits_{j=1}^n(a_{\pi(j)}/\sum\limits_{k=j}^nb_{\pi(k)})\label{eqt3}
\end{equation}

\section{ARP is NP-complete}
\label{sec:2}

\begin{theorem}\label{thm:1} ARP is NP-complete.\end{theorem}
\begin{proof}It is easy to see that ARP is in $\mathcal{NP}$: The certificate could be a permutation of the airplanes, and a certifier checks that the order contains each airplane exactly once and the length of the corresponding flight is greater than $S$.

We now show that NVEP $\leq_P$ ARP. Given an NVEP instance with $n$ vehicles, we define an ARP instance with $n$ airplanes. Let $v_i = a_i$ and $c_i = b_i$ for each $1 \leq i \leq n$. We claim that there is a running order $\pi$ for the NVEP instance such that $D_{\pi} \geq n$, if and only if there is a sequence of length that is no less than $n$ in our ARP instance.

If the NVEP has a running order such that $D_{\pi} \geq S$, then this ordering of the corresponding airplanes defines a running order of length greater than $S$. Conversely, suppose there is a sequence of length $S_{\pi}$ that is greater than $S$, the corresponding running order of length must be greater than $S$.\end{proof}

\section{Main result}
\label{sec:3}

\qquad We proposed the definition of the sequential feasible solution \cite{Cui23P}, which has following characteristics: if an airplane refueling problem has feasible solutions, it must have sequential feasible solutions, and its optimal feasible solution must be the optimal sequential feasible solution. It means that an algorithm by enumerating all the sequential feasible solutions must be an optimal algorithm for ARP, then an sequential search algorithm was proposed, ant its computational complexity depends on the number of sequential feasible solutions referred to $Q_n$. It is shown that $Q_n \leq 2^{n-2}$ for worst case. However, we found that $Q_n$ is no longer exponential function when the input scale is large enough. Then we prove that $Q_n$ is bounded by $\frac{m^2}{n}C_n^m$ when $n$ is greater than $2m$. Here $m$ is a constant and $2m$ is regarded as the "inflection point" of the complexity of the sequential search algorithm from exponential time to polynomial time.

In Book \cite{garey79} (Figure 2.6), $\mathcal{NP}$ is divided into "the land of $\mathcal{P}$" and "the land of NP-complete". If $\mathcal{P}$ differs from $\mathcal{NP}$, then there must exist problems in $\mathcal{NP}$ that are neither solvable in polynomial time nor NP-complete. We have shown the  NP-completeness of ARP in Theorem \ref{thm:1}, we now show that ARP is solvable in polynomial time in Theorem \ref{thm:2}.

\begin{theorem}\label{thm:2} \cite{Cui23P} For large scale of $n$-airplane instances with $v_1 / c_1^2 > \cdots > v_n / c_n^2$ and $v_1 / c_1 < \cdots < v_n / c_n$ under the assumptions of $v_n / c_n \leq M$. There exist an index $m$ and $N = 2m$, such that $Q_n < \frac{m^2}{n}C_n^m$ for any $n > N$. In addition, when $n > N$, $Q_n^{(m)}$ is upper bounded by $n^m$, and the ARP is solvable in polynomial time.\end{theorem}

\begin{corollary}\label{cor:1} \cite{Cui23P} Suppose $\mathcal{A}$ is a set of $n$ airplanes with $S_A \leq v_1 / c_1 < \cdots < v_n / c_n \leq S_B$. There must be an index $\bar{m}$ related to $\mathcal{A}$, such that for any ARP instance chosen from $\mathcal{A}$, its "inflection point" related $2m$ must be less than $2\bar{m}$.\end{corollary}

\begin{figure}[ht]
  \includegraphics[width=0.75\textwidth]{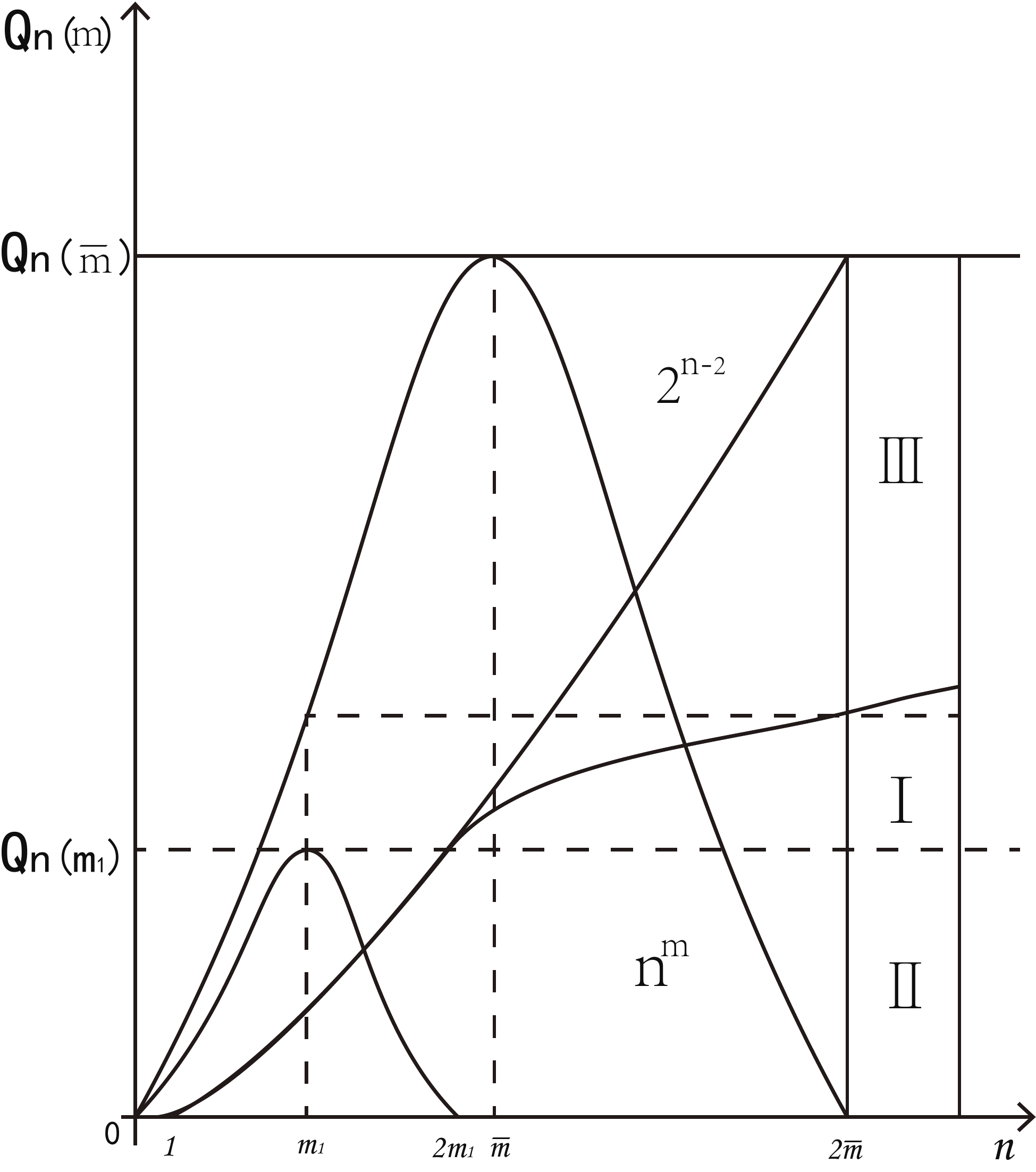}
  \caption{The ARP is NP-complete and is solvable in polynomial time. $Q_n(m) \leq 2^{n - 2}$ when $n \leq 2m$, and $Q_n(m) \leq n^m$ when $n$ is greater than $2m$. When $m = m_1$, $Q_n(m_1) \leq 2^{n - 2}$ for $n \leq 2m_1$, and $Q_n(m_1) < n^{m_1}$ when $n > 2m_1$. For a given set of airplanes, $\bar{m}$ is the upper bound of $m$ according to Corollary \ref{cor:2}. Thus $Q_n(\bar{m}) \leq 2^{n - 2}$ for $n \leq 2\bar{m}$, and $Q_n(\bar{m}) < n^{\bar{m}}$ when $n > 2\bar{m}$.}
  \label{fig:P/NP}
\end{figure}

The main idea is shown in Fig. \ref{fig:P/NP}. When $n$ is in small size, the sequential search algorithm runs in exponential time. However, when $n$ gets larger than a particular number, the cumulative sums of fuel consumption rates referred to $C$ by (\ref{eqt2}) is monotonically increasing, and the maximal flight length of single airplane is upper bounded. There must exist an index $m$, according to Theorem \ref{thm:2}, such that when $n > 2m$, the growth rate of the number of sequential feasible solutions is getting slowing down to a polynomial time. Since the maximal flight length of single airplane is limited, by Theorem \ref{cor:1}, there exist an index $\bar{m}$, which is referred to an upper bound of $m$. When $n > 2 \bar{m}$, the given ARP instance is solvable in polynomial time. Therefore, ARP is NP-complete and is solvable in polynomial time.

\begin{theorem}\label{thm:3} Let X be an NP-complete problem, then X is solvable in polynomial time if and only if $\mathcal{NP} = \mathcal{P}$  holds. (see Theorem 8.12 in Book \cite{tardos06})\end{theorem}

\section{Conclusion}
\label{sec:4}

\qquad It is proved that the ARP is NP-complete. In Article \cite{Cui23P}, ARP is proved to be solvable in polynomial time. According to Theorem \ref{thm:3}, we find such a problem as ARP, that is both NP-complete and is solvable in polynomial time.

\bibliography{references}

\begin{thebibliography}{15}
\providecommand{\natexlab}[1]{#1}
\providecommand{\url}[1]{{#1}}
\providecommand{\urlprefix}{URL }
\providecommand{\doi}[1]{\url{https://doi.org/#1}}
\providecommand{\eprint}[2][]{\url{#2}}
 \bibcommenthead

\bibitem[{Cui and Li(2023{\natexlab{a}})}]{Cui23NP}
Cui J, Li X (2023{\natexlab{a}}) The $n$-vehicle exploration problem is
  np-complete. arXiv:2304.03965v1, \doi{10.48550/arXiv.2304.03965}

\bibitem[{Cui and Li(2023{\natexlab{b}})}]{Cui23P}
Cui J, Li X (2023{\natexlab{b}}) A polynomial-time algorithm to solve the
  aircraft refueling problem: the sequential search algorithma polynomial-time
  algorithm to solve the aircraft refueling problem: the sequential search
  algorithm. arXiv:2210.11634v2, \doi{10.48550/arXiv.2210.11634}

\bibitem[{Gamow and Stern(1958)}]{puzzle58}
Gamow G, Stern M (1958) Puzzle-Math. Viking Press, New York

\bibitem[{H{\"o}hn(2013)}]{hohn13}
H{\"o}hn W (2013) Complexity of generalized min-sum scheduling. Scheduling:
  Dagstuhl Seminar Reports 13111, \urlprefix\url{http://www.dagstuhl.de/13111}

\bibitem[{Iftah and Danny(2019)}]{gamzu19}
Iftah G, Danny S (2019) A polynomial-time approximateion shceme for the
  airplane refueling problem. J Sched 22:429--444.
  \doi{10.1007/s10951-018-0569-x}

\bibitem[{Kleinberg and Tardos(2006)}]{tardos06}
Kleinberg J, Tardos {\'E} (2006) Algorithm design. Pearson Education Asia
  Limited and Tsinghua University Press, Cornell University

\bibitem[{Li et~al(2019)Li, Hu, Luo, and Cui}]{lijs19}
Li J, Hu X, Luo J, et~al (2019) A fast exact algorithm for airplane refueling
  problem. Combinatorial Optimization and Applications, Proceedings LNCS 11949,
  Springer Internaticonal Publishing:316--327.
  \doi{10.1007/978-3-030-36412-0_25}

\bibitem[{Li and Cui(2009)}]{lixy09}
Li X, Cui J (2009) Real-time algorithm scheme for $n$-vehicle exploration
  problem. Combinatorial optimization and applications COCOA 2009, HuangShan,
  China, June 10-12, 2009 Proceedings Berlin: Springer LNCS 5573, Springer,
  Berlin, Heidelberg:287--300. \doi{10.1007/978-3-642-02026-1_27}

\bibitem[{Michael R.~Garey(1979)}]{garey79}
Michael R.~Garey DSJ (1979) Computers and intractability: a guide to the theory
  of NP-Completeness. W. H. Freeman and Company, New York

\bibitem[{Pereira and V{\'a}squez(2017)}]{vasquez17}
Pereira J, V{\'a}squez OC (2017) The single machine weighted mean squared
  deviation problem. European J Oper Res 2(261):515--529.
  \doi{10.1016/j.ejor.2017.03.001}

\bibitem[{V{\'a}squez(2015)}]{vasquez15}
V{\'a}squez OC (2015) For the airplane refueling problem local precedence
  implies global precedence. Optim Lett 9:663--675.
  \doi{10.1007/s11590-014-0758-2}

\bibitem[{Woeginger(2010)}]{woeginger10}
Woeginger GJ (2010) The airplane refueling problem. Scheduling: Dagstuhl
  Seminar Proceedings 10071,
  \urlprefix\url{http://drops.dagstuhl.de/opus/volltexte/2010/2536}

\bibitem[{Yang(1999)}]{shuxue}
Yang R (1999) A travelling problem and its extended research. Bulletin of Maths
  9:44--45. \doi{CNKI:SUN:SXTB.0.1999-09-025}

\bibitem[{Yu and Cui(2018)}]{yu18}
Yu F, Cui J (2018) Research on the efficient computation mechanism - in the
  case of n-vehicle exploration problem. Acta Math Appl Sin Engl Ser
  34(3):645--657. \doi{10.1007/s10255-018-0774-6}

\bibitem[{Zhang and Cui(2021)}]{zhang21}
Zhang G, Cui J (2021) A novel milp model for n-vehicle exploration problem. J
  Oper Res Soc China 9:359--373. \doi{10.1007/s40305-019-00289-2}

\end{thebibliography}

\end{document}